\documentclass[11pt, onecolumn, draftcls, peerreview]{IEEEtran}
\usepackage{cite, amssymb,booktabs,multirow}
\usepackage{epsfig}
\usepackage{epstopdf}
\usepackage{enumerate, morefloats}
\usepackage{todonotes}
\usepackage{ifthen}
\usepackage[cmex10]{amsmath}
\usepackage{amsthm}
\usepackage{cleveref}
\newtheorem{Lemma}{Lemma}
\newtheorem{Theorem}{Theorem}
\def\ci{\perp\!\!\!\perp}
\renewcommand{\d}[1]{\ensuremath{\operatorname{d}\!{#1}}}
\begin{document}
\author{Reza~Parseh,~\IEEEmembership{Student Member, IEEE,}~and Kimmo Kansanen,~\IEEEmembership{Member, IEEE}\thanks{© 2014 IEEE. Personal use of this material is permitted. Permission from IEEE must be obtained for all other uses, in any current or future media, including reprinting/republishing this material for advertising or promotional purposes, creating new collective works, for resale or redistribution to servers or lists, or reuse of any copyrighted component of this work in other works.

This article has been accepted for publication in IEEE Trans. Signal Process. DOI 10.1109/TSP.2014.2360820 IEEE Transactions on Signal Processing. 

The authors are with the Department of Electronics and Telecommunications, Norwegian University of Science and Technology, 7034, Trondheim, Norway (email: reza.parseh@iet.ntnu.no; kimmo.kansanen@iet.ntnu.no).}}

\title{On Estimation Error Outage for Scalar Gauss-Markov Signals Sent Over Fading Channels}
\maketitle

\begin{abstract}
Measurements of a scalar linear Gauss-Markov process are sent over a fading channel. The fading channel is modeled as independent and identically distributed random variables with known realization at the receiver. The optimal estimator at the receiver is the Kalman filter. In contrast to the classical Kalman filter theory, given a random channel, the Kalman gain and the error covariance become random. Then the probability distribution function of expected estimation error and its outage probability can be chosen for estimation quality assessment. In this paper and in order to get the estimation error outage, we provide means to characterize the stationary probability density function of the random expected estimation error. Furthermore and for the particular case of the i.i.d. Rayleigh fading channels, upper and lower bounds for the outage probability are derived which provide insight and simpler means for design purposes. We also show that the bounds are tight for the high SNR regime, and that the outage probability decreases linearly with the inverse of the average channel SNR. 
\end{abstract}
\begin {IEEEkeywords} Estimation Over Fading Channels, Kalman Filter, Outage Probability, Uncoded Transmission
\end{IEEEkeywords}
\section{Introduction}\label{Sec_I} 
Low or zero delay transmission of measurements of a dynamic system to a remote controller/observer is important in applications such as network monitoring and control, wireless sensor networks, and generally in real-time signal processing when the observed signals should be sent over a communication channel. Due to tight delay conditions in many cases, high-performance block-wise coding schemes which incur unacceptable delay should be avoided. For wireless fading channels, it is possible to send the measurements directly over the channel using uncoded transmission and then perform estimation on the channel outputs at the receiver. Analysis of the signal estimation quality is therefore necessary to ensure satisfactory performance for such applications.    

The literature for network communication and control and wireless sensor networks is diverse and rich (see \cite{zhang2001stability, lihua2011control, antsaklis2007special,Stankovic:2008hk} and the references within). For various applications where the dynamic system follows a Gauss-Markov model and the channel realization is independent of the randomness of the dynamic system, the optimal estimator is the Kalman filter (\cite{bougerol1993kalman,shi2009kalman, Shi:2012en, rohr2011kalman, you2011mean, quevedo2012kalman, Subramanian:2005ca, Zhu:2009dg}). Due to the randomness of the fading channel, the Kalman filter is random and does not necessarily converge to a constant value. The instantaneous estimation error covariance is random as well. The estimation error covariance matrix is related to the prediction error covariance matrix with a simple transform. The prediction error covariance matrix propagates through a Riccati equation studied extensively in the literature. With the channel matrices being random, the prediction error covariance matrices then constitute a well-known stochastic process referred to as the random Riccati equation (RRE)\cite{wonham1968matrix}. While we focus on the estimation error quality as we are interested in the signal reconstruction, others especially in the control literature have focused on the prediction error covariance matrix because it is used directly in the controller design in many cases. We review some of those works in the following.  

In \cite{wang1999stability}, stability of RRE is studied and it is shown that under mild assumptions on the random observability Gramian matrix, it is both $L_r$ and exponentially stable. In \cite{xie2008stability}, the peak covariance stability of the RRE resulting from Kalman filtering with random observation losses is studied. Boundedness of the covariance matrix in the usual sense is also considered in the same work. In \cite{Moayedi:2010tV}, an adaptive filtering scheme based on the Riccati equation is proposed for state estimation in network control systems subject to delays, packet drops and missing measurements. In \cite{kar2012kalman}, it was shown that sequence of random covariance matrices converges in probability when observations are sent over a packet erasure channel where the erasure event is a Bernoulli i.i.d process. The stationary distributions for infinitely large random matrices with good approximations for dimensions around 10-20  were also found in \cite{vakili2008stieltjes} and \cite{vakili2008riccati} for two classes of random Riccati and Lyapunov equations. Also in \cite{Dey2009kalman} bounds on the mean of the instantaneous covariance matrices in the RRE formulation are obtained. 

In this paper, we study the case when measurements of a scalar Gauss-Markov process are sent over a fading channel with i.i.d. channel realizations. This model best suits e.g. low-cost sensor networks with processing at the fusion center. The samples are sent over the channel using continuous-valued uncoded (also known as analog \cite{caire2007distortion}) transmission after they are obtained, to avoid processing at the sensor or block coding and the consequent delays. It is assumed that the full channel knowledge is available at the receiver at the time of the observation. The optimum MMSE filter, i.e. the Kalman filter is then random and the exact value of the instantaneous estimation error variance (IEV) cannot be obtained in advance. For that reason, we are in particular interested in statistical characterization of the resulting estimation error. 

In the spirit of outage in fading channels, we utilize estimation error outage as a criterion for estimation performance assessment. A similar property, namely distortion outage was proposed in \cite{peng2010distortion} for MIMO block fading channels from an information theoretical point of view. There, it is mentioned that as for our case, distortion outage measures are useful when delay is of concern. Outage is defined as the event where the IEV exceeds a certain threshold. From a more practical viewpoint, this measure could be used as a design parameter for a control or monitoring system which observes the process. While \cite{peng2010distortion} finds expressions for the distortion outage diversity order and proves achievability, we are interested in the practical case of the Kalman filter and its behavior. We try to find the estimation error outage probability and find out how it is related to average channel SNR under certain channel statistics. We show for instance that in the scalar case and for the i.i.d. Rayleigh fading channel, the outage measure takes on a simple form in the high SNR regime, which we believe is insightful for design purposes and further development. 

In the rest of the paper and after introducing the system model in details, we first show that for \textit{any} i.i.d. fading channel, the first order probability density function (pdf) may be obtained through a recursive integral equation. We then select the case of Rayleigh fading channel and provide upper and lower bounds for the outage probability. Next, we show that the bounds are tight for the high SNR regime. Finally, we show that the outage probability decreases linearly with inverse of the channel SNR in the high SNR regime. This could for instance be used a rule of thumb method for estimation quality assessment under settings discussed in this paper. The high SNR analysis enables us to perform diversity analysis for the Kalman estimator as well. A shorter version of this work without high SNR analysis appears in \cite{parseh2014dist}. The current work also provides a different and more straightforward proof for the integral equation characterizing the pdf of the IEV. 
\section{System Model and Problem Definition}\label{Sec_II}
Consider the following scalar complex Gauss-Markov signal model 
\begin{align}
\nonumber x(n) & = \rho x(n-1) + u(n), \, n\geq 1,\, x(0)\sim\mathcal{CN}(0,M(0))\\
y(n) & = h(n) x(n) + v(n),
\end{align}
with $u(n)$ and $v(n)$ as white circularly symmetric complex Gaussian random variables with variances $\sigma^2_u > 0$ and $\sigma^2_v > 0$, respectively. Consider $h(n)$ to be i.i.d. samples of a random variable (starting from Sec. \ref{Sec_III_B}, we will assume that channel is Rayleigh fading). We also assume that $h(n)$ cannot be equal to zero for all $n$ (non-existent channel is not included). This signal model characterizes e.g. measurements of a first-order Gauss-Markov process sent over a fading channel. It is assumed that perfect knowledge of the random channel $h(n)$ is available at the receiver and $h(n)$ are also independent of $u(n)$ and $v(n)$. For further development in Sec. \ref{Sec_III_A}, we require that $h(1) \neq 0$ and $\rho \neq 0$. The objective at the receiver is optimal estimation of the signal $x(n)$, given the channel outputs. 

Given the previous assumptions, and with $h(n)$ independent of $u(n)$ and $v(n)$, the optimal MMSE estimator of ${x}(n)$ based on the observations $y(n)$ is the well-known Kalman filter \cite{kailath2000linear} with the following steps
\begin{align}
& \label{eqIII_1}\hat{x}(n | n-1) = \rho \hat{x}(n-1 | n-1)\\
& \label{eqIII_2} P(n) = \rho^2 M(n-1) +\sigma^2_u\\
&  \label{eqIII_3} K(n) = P(n) h^*(n) [\sigma^2_v + |h(n)|^2 P(n)]^{-1}\\
& \label{eqIII_4} \hat{x}(n | n) = \hat{x}(n | n-1) + K(n) (y(n) - h(n) \hat{x}(n | n-1))\\
& \label{eqIII_5} M(n) = \left(1 - K(n) h(n)\right) P(n).
\end{align}
Concisely stated, eq. \eqref{eqIII_1} is the prediction of the current state based on the previous estimated state (\textit{a priori estimate}) using the system model and eq. \eqref{eqIII_2} is the instantaneous expected (with respect to noise) prediction error. Equation \eqref{eqIII_3} is the corresponding Kalman gain equation and eq. \eqref{eqIII_4} is the correction equation based on the Kalman gain update (\textit{a posteriori estimate}). Finally eq. \eqref{eqIII_5} provides us with the instantaneous estimation error variance.  

It is straightforward to show that both the $P(n)$, i.e. the prediction error variance and $M(n)$, i.e. the estimation error variance may be written recursively in terms of their previous values and current value of $h(n)$, where one is a deterministic function of the other. The statistical properties of the one may then be acquired using the statistical properties of the other. In the rest of this paper, we study $M(n)$. 

The recursion for $M(n)$ is obtained  as follows
\begin{align}
 M(n) & = \left(1 - K(n) h(n)\right) P(n)\nonumber\\
& = \left (1 - \frac{P(n) |h(n)|^2}{\sigma^2_v + |h(n)|^2 P(n)}\right) P(n)\nonumber\\
 & = P(n) - \frac{P^2(n) |h(n)|^2}{\sigma^2_v + |h(n)|^2 P(n)}\nonumber\\
& = \frac{P(n)\sigma^2_v}{\sigma^2_v + |h(n)|^2 P(n)}\nonumber\\
& = \frac{P(n)}{1 + |h(n)|^2 P(n) / \sigma^2_v}\nonumber\\
& \stackrel{(a)}= \frac{\rho^2 M(n-1) +\sigma^2_u}{1 + \gamma(n) \left(\rho^2 M(n-1) +\sigma^2_u \right)},\label{eq_15}
\end{align}
where in $(a)$ we have set $\gamma(n) = |h(n)|^2 / \sigma^2_v$ to simplify representation of the functions which depend on the channel. With this notation, $\gamma(n)$ corresponds to the instantaneous SNR at the receiver side. 

In order to characterize the random estimation outage event, we define estimation error outage probability (EOP) as
\begin{align}
P^n_\textnormal{out}({M_\textnormal{th}})= \textnormal{Pr}(M(n) \geq M_\textnormal{th} )
\end{align}
and in particular the asymptotic EOP which is of interest, in order to characterize the steady-state distributions, i.e. 
\begin{align}
P_\textnormal{out}(M_\textnormal{th})= \lim_{n\rightarrow \infty} P_\textnormal{out}^n({M_\textnormal{th}}) = \lim_{n\rightarrow \infty} \textnormal{Pr}(M(n) \geq M_{\textnormal{th}}).
\end{align}
Clearly $P^n_\textnormal{out}({M_\textnormal{th}}) = 1 - F_{M(n)}(M_\textnormal{th})$ and $P_\textnormal{out}(M_\textnormal{th}) = 1 - F_M(M_\textnormal{th})$,
where $F_{M(n)}(M)$ $\left(F_M(M)\right)$ is the  (steady state) cumulative distribution function (cdf) of $M(n)$. 
\section{Statistical Properties of Instantaneous Estimation Error Variance}\label{Sec_III}
In this section, we study the asymptotic probability density function of the IEV. Because $F_M(M)$ and $f_M(M)$ are related with a linear operation (derivative), we begin to study $f_M(M)$. After that, the EOP will readily be obtained with one integration operation.
\subsection{Asymptotic Pdf of The Instantaneous Estimation Error Variance}\label{Sec_III_A}
Given \eqref{eq_15}, it is easy to verify that for any arbitrary positive real number $M$, $M(n) \leqslant M$ leads to 
\begin{align*}
\gamma(n) \geqslant \frac{1}{M} - \frac{1}{\rho^2 M(n-1) + \sigma^2_u}.
\end{align*} 
Also, we have that $\gamma(n) \geqslant 0$ and $0 < M(n) < M_\textnormal{max}$, where $M_\textnormal{max}$, the upper limit for $M(n)$ is obtained from 
\begin{align}
M_\textnormal{max} = \left \{ \begin{array}{ll} \infty, & |\rho|\geqslant 1\\
\frac{\sigma^2_u}{1-\rho^2}, & |\rho| < 1 .
\end{array}\right.
\end{align} 
$M_\textnormal{max}$ is effectively the estimation error variance for the worst channel, i.e. $\gamma(n) = 0$ with probability 1. In that case, the best estimator is the average mean, i.e. $\hat{x}(n) = E(x(n)) = 0$ and therefore the estimation error variance is equal to $M_\textnormal{max}$. Also note that the case $|\rho| \geqslant 1$ is of little practical importance in our case, because for a divergent signal, continuous-amplitude uncoded transmission would not be practical. It is however included in the definition of $M_\textnormal{max}$ to show that the analysis does not depend on the value of $\rho$. 
 
Given the above limits and conditions for $\gamma(n)$ and $M(n)$ and according to \cite{papoulis2002probability}, it is possible to get the cdf of $M(n)$, i.e. $F_{M(n)}(M)$ as follows. 
\begin{align}
 F_{M(n)}(M) = \int_{0}^{M_\textnormal{max}}\int_{\frac{1}{M} - \frac{1}{\rho^2 M(n-1) + \sigma^2_u}}^{\infty} f_{\gamma(n), M(n-1)}\left(\gamma(n), M(n-1)\right) d\gamma(n) \d M(n-1),
\end{align}
where $ f_{\gamma(n), M(n-1)}\left(\gamma(n), M(n-1)\right)$ is the joint pdf of $\gamma(n)$ and $M(n-1)$. The pdf for $M(n)$ is then obtained by simply applying $f_{M(n)}(M) = \frac{\partial}{\partial M} F_{M(n)}(M)$. That leads to 
\begin{align}
 f_{M(n)}(M) & = \int_{0}^{M_\textnormal{max}} \frac{\partial}{\partial M} \int_{\frac{1}{M} - \frac{1}{\rho^2 M(n-1) + \sigma^2_u}}^{\infty} 
 f_{\gamma(n), M(n-1)}\left(\gamma(n), M(n-1)\right) d\gamma(n) \d M(n-1)\nonumber\\
& = \int_{0}^{M_\textnormal{max}} \frac{1}{M^2}  f_{\gamma(n), M(n-1)}\left(\frac{1}{M} - \frac{1}{\rho^2 M(n-1) + \sigma^2_u}, M(n-1)\right) \d M(n-1),
\end{align}
or with some change of notation, 
\begin{align}
 f_{M(n)}(M) = \int_{0}^{M_\textnormal{max}} \frac{1}{M^2} f_{\gamma(n), M(n-1)}\left(\frac{1}{M} - \frac{1}{\rho^2 m + \sigma^2_u}, m\right) \d m.\label{eqn_f_M}
\end{align}
Now if we assume that $\gamma(n)$ is independent of $M(n-1)$ \big($\gamma(n) \ci M(n-1)$\big), we may rewrite \eqref{eqn_f_M} as 
\begin{align}
 f_{M(n)}(M) =  \frac{1}{M^2} \int_{0}^{M_\textnormal{max}} f_{\gamma(n)}\left(\frac{1}{M} - \frac{1}{\rho^2 m + \sigma^2_u}\right) f_{M(n-1)} (m) \d m.\label{eqn_f_M_II}
\end{align}
Note that as we have assumed i.i.d. channels, then we have that $\gamma(n) \ci \gamma (i)$ for $i < n$. We can simply assume that $\gamma(i) \ci M(0)$ ($M(0)$ is a constant). As a result, we obtain that $\gamma(n) \ci M(n-1)$ because $M(n-1)$ is a function of $M(0)$ and $\gamma(1), \gamma(2), \cdots, \gamma(n-1)$ only. Thus i.i.d. channel assumption is a sufficient condition to get the main result in \eqref{eqn_f_M_II}.
As we have assumed an i.i.d. channel, all $\gamma(n)$ have the same pdf, i.e. $f_{\gamma(n)}(\gamma(n)) = f_\gamma(\gamma(n))$. Therefore, we rewrite \eqref{eqn_f_M_II} to obtain
\begin{align}
 f_{M(n)}(M) = \frac{1}{M^2} \int_{0}^{M_\textnormal{max}} f_{\gamma}\left(\frac{1}{M} - \frac{1}{\rho^2 m + \sigma^2_u}\right) f_{M(n-1)} (m) \d m.\label{eqn_f_M_II_a}
\end{align}
Equation \eqref{eqn_f_M_II_a} may be used to find $f_{M(n)}(M)$ for any $n$ by starting from $f_{M(0)}(M) = \delta(M - M(0))$ and iterating over $n$. However, the objective of this paper is outage analysis and for that purpose, we need the steady-state distribution. In the following, we present Theorem \ref{theorem:I} which proves the existence of a steady-state distribution for $M(n)$, namely $f_M(M)$ and outlines how it can be obtained.    
\begin{Theorem}\label{theorem:I}
The random process $M(n)$ converges in law and has a steady-state distribution, namely $f_M(M)$ which satisfies the following equality
\begin{align}
 f_{M}(M) = \frac{1}{M^2} \int_{0}^{M_\textnormal{max}} f_{\gamma}\left(\frac{1}{M} - \frac{1}{\rho^2 m + \sigma^2_u}\right) f_{M} (m) \d m.\label{eqn_f_M_III}
\end{align}
\end{Theorem}
\begin{proof}
See Appendix \ref{app_1}.
\end{proof}
After stating Theorem \ref{theorem:I}, we will utilize \eqref{eqn_f_M_III} for the rest of the analysis in order to characterize $f_M(M)$. To be more specific with \eqref{eqn_f_M_III}, we use the fact that $\gamma(n) \geqslant 0$. That necessitates that the argument of the function $f_\gamma()$ should always be positive. Clearly, if $M \leqslant \sigma^2_u$, the term $\frac{1}{M} - \frac{1}{\rho^2 m + \sigma^2_u}$ is always positive. However for $M > \sigma^2_u$, the integral should be taken over the range of $\gamma$ where $\frac{1}{M} - \frac{1}{\rho^2 m + \sigma^2_u} \geqslant 0$, i.e. for $m\geqslant \frac{M-\sigma^2_u}{\rho^2}$. With this background, we can finally provide the following lemma that describes the asymptotic pdf of $M(n)$, i.e. $f_M(M)$ in terms of itself integrated with a kernel which is a function of the instantaneous channel SNR. Solving this equation leads to $f_M(M)$ and with one integration to $P_\textnormal{out}$, which is the target.   
\begin{Lemma}\label{lemma:I}
Asymptotic pdf of $M(n)$, i.e. $f_M(M)$ can be obtained from the following equation
\begin{align}
 f_M(M) \label{eqn_main_her}= \left\{
\begin{array}{ll}
 \frac{1}{M^2}\int_{0}^{M_\textnormal{max}} f_{\gamma}\left( \frac{1}{M} - \frac{1}{\rho^2 m +\sigma^2_u}\right)f_M(m)\d m, & M\leqslant \sigma^2_u\\
 \frac{1}{M^2}\int_{\frac{M-\sigma^2_u}{\rho^2}}^{M_\textnormal{max}} f_{\gamma}\left( \frac{1}{M} - \frac{1}{\rho^2 m +\sigma^2_u}\right)f_M(m)\d m, & M> \sigma^2_u.
\end{array} \right.
\end{align}
\end{Lemma}
The solution is general and is explicitly given in terms of instantaneous channel SNR pdf and system parameters. Though \eqref{eqn_main_her} can be solved numerically if needed, the general closed-form solution does not seem to be readily attainable. In the following, we have focused on the important case of Rayleigh fading channels where $f_\gamma(\gamma) = \lambda \textnormal{e}^{-\lambda \gamma} \mathcal{U}(\gamma)$ \big($\mathcal{U}()$ is the unit step function\big). Note that with this definition, $\lambda = 1 /  E(\gamma(n)) = \sigma^2_v /  E(|h(n)|^2)$, i.e. stronger channels yield smaller values for $\lambda$ and vice versa.   
\begin{figure}
\centering
\includegraphics[scale = 0.65] {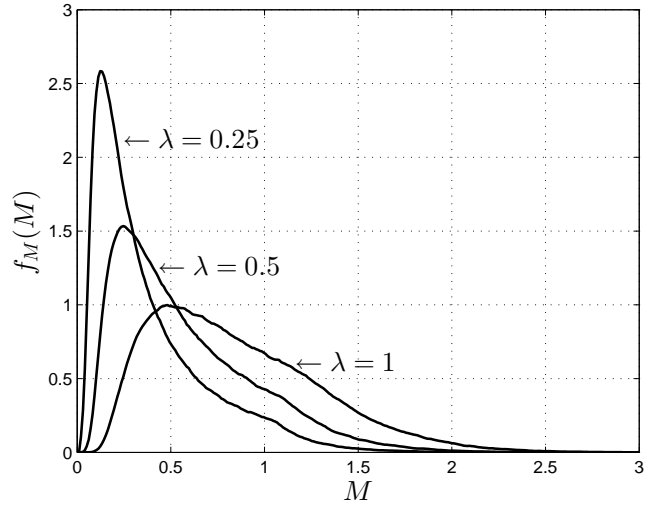}
\caption{Asymptotic pdf of $M(n)$ for $\sigma^2_u = \sigma^2_v = 1$, $\rho = 0.95$, $\lambda = 1, 0.5, 0.25$ (SNR = 0, 3, 6 dB respectively). Note the break point at $M = \sigma^2_u$.}
\label{fig_1}
\end{figure}
\subsection{Pdf of The Instantaneous Estimation Error Variance Under Rayleigh Fading}\label{Sec_III_B}
We can rewrite \eqref{eqn_main_her} given that channel is i.i.d. Rayleigh fading. Using that we obtain
\begin{align}
 f_M(M) =  \frac{\lambda}{M^2} \textnormal{exp}(\frac{-\lambda}{M}) \left\{
\begin{array}{ll}
 \int_{0}^{M_\textnormal{max}} \textnormal{exp}\left(\frac{\lambda}{\rho^2 m +\sigma^2_u}\right)f_M(m)\d m, & M\leqslant \sigma^2_u\\
 \int_{\frac{M-\sigma^2_u}{\rho^2}}^{M_\textnormal{max}} \textnormal{exp}\left(\frac{\lambda}{\rho^2 m +\sigma^2_u}\right)f_M(m)\d m, & M> \sigma^2_u,
\end{array} \right.
\label{eqn_main}\end{align}
which in order to get more insight and with some algebraic manipulation, can also be written as 
\begin{align}
 f_M(M)  = \frac{\lambda}{M^2}\textnormal{exp}(\frac{-\lambda}{M}) \left\{
\begin{array}{ll}
 \kappa , & M\leqslant \sigma^2_u\\
 \kappa
-\int_{0}^{\frac{M-\sigma^2_u}{\rho^2}} \textnormal{exp}\left(\frac{\lambda}{\rho^2 m +\sigma^2_u}\right)f_M(m)\d m, & M> \sigma^2_u,
\end{array} \right.
\label{eqn_main_2}\end{align}
where
\begin{align}
\label{eqn_kappa}\kappa = \int_{0}^{M_\textnormal{max}}\textnormal{exp}\left(\frac{\lambda}{\rho^2 m +\sigma^2_u}\right)f_M(m)\d m.
\end{align}
Though in general $\kappa$ depends on the pdf itself, \eqref{eqn_main_2} is insightful in the sense that it shows the exact shape of the pdf for the first part where $M \leqslant \sigma^2_u$.  

Typical shapes of such pdf's which are obtained through Monte-Carlo simulations are depicted in Fig. \ref{fig_1} to further highlight the points mentioned. For Fig. \ref{fig_1}, it is assumed that $\sigma^2_u=\sigma^2_v =1$, $\lambda = 1, 0.5, 0.25$ (SNR = 0, 3, 6 dB respectively), and $\rho = 0.95$. Note that the pdf support is theoretically bounded in this case at point $M_\textnormal{max} = \sigma^2_u / (1 - \rho^2) \approx 10.26$ (not shown in the figure due to its insignificance). 
\begin{figure}[t]
\centering
\includegraphics[scale = 0.65] {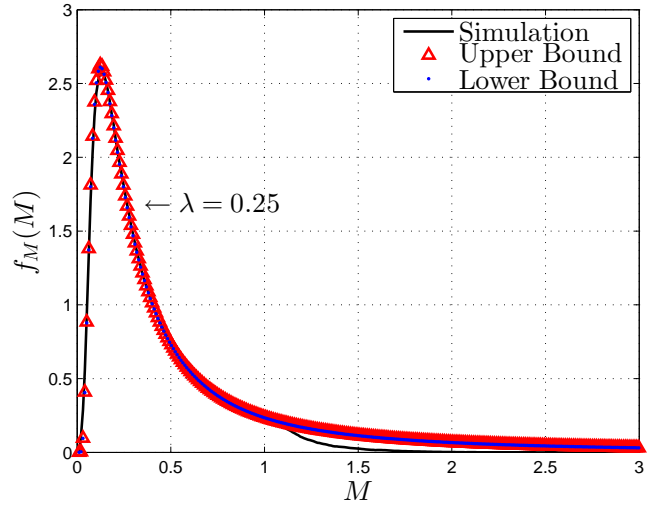}
\caption{Asymptotic pdf of $M(n)$ and its approximates using upper and lower bounds for $\kappa$ given that $\sigma^2_u = \sigma^2_v = 1$, $\rho = 0.95$, $\lambda = 0.25$ (SNR = 6 dB).}
\label{fig_2}
\end{figure}
Also note that the break point, $M = \sigma^2_u$ where the pdf changes shape is quite visible in Fig. \ref{fig_1}. Also, from \eqref{eqn_main_2}, it is easily verified that the pdf has an extremum point at $M = \lambda /2$ for the given SNR values. This extremum point is only a function of the average channel SNR, i.e. $ E(\gamma(n)) = 1/\lambda$.  

The break point $M = \sigma^2_u$ corresponds to the steady-state variance of the signal when there is no correlation ($\rho = 0$), whereas the point $M_\textnormal{max}  = \sigma^2_u / (1 - \rho^2)$ corresponds to the upper limit for the support of $f_M(M)$ (maximum value for the IEV) for the worst channel ($\gamma(n) = 0$) when no information gets passed the channel and the estimator is equal to $\hat{x}(n) = E(x(n)) = 0$. It is quite visible and theoretically verifiable that the pdf tail vanishes rapidly after the break point if the SNR increases. Also that the higher the threshold, the lower the outage probability would be. As a result, getting bounds on the first part helps with understanding the pdf behavior better and at the same time get approximate values and bounds for $P_\textnormal{out}(M_\textnormal{th})$. Using \eqref{eqn_main_2} and \eqref{eqn_kappa}, we find upper and lower bounds for $\kappa$, approximations for the pdf and upper and lower bounds for the outage for the first part ($M\leqslant \sigma^2_u$). Another insight from \eqref{eqn_kappa} is that the pdf shape is independent of whether the system is stable ($\rho < 1$) or unstable ($\rho \geqslant 1$), though the value of $\kappa$ depends on $\rho$. 

Though the pdf is given by the equation $f_M(M) = \frac{\kappa  \lambda}{M^2} \textnormal{exp}(\frac{-\lambda}{M}) \, (M\leqslant \sigma^2_u)$, the exact value of $\kappa$ depends on the whole pdf and cannot be known without solving \eqref{eqn_main_2}. However, it is possible to obtain the following bounds for $\kappa$, namely $\kappa_{l} < \kappa <\kappa_{u}$, which later on are used to characterize two functions $P_\textnormal{out}^l(M_\textnormal{th})$ and $P_\textnormal{out}^u(M_\textnormal{th})$ for which $P_\textnormal{out}^l(M_\textnormal{th}) < P_\textnormal{out}(M_\textnormal{th}) < P_\textnormal{out}^u(M_\textnormal{th})$ for all $M \leqslant \sigma^2_u$.  

\begin{Lemma}\label{lemma:II}
For all $M \leqslant \sigma^2_u$, we have $\kappa_{l} < \kappa <\kappa_{u}$, where
\begin{align}
\kappa_{u} & =\frac{1}{\left( a_{\kappa} \textnormal{exp}\left(\frac{-\lambda}{\sigma^2_u(1+\rho^2)}\right) + \textnormal{exp}(-\frac{\lambda}{\sigma^2_u})\right) }\\
\kappa_{l} & =\frac{1}{\left( a_{\kappa} \textnormal{exp}(\frac{-\lambda}{\rho^2 M_\textnormal{max} + \sigma^2_u}) + \textnormal{exp}(-\frac{\lambda}{\sigma^2_u})\right) },
\end{align}
where we have defined 
\begin{align}
a_{\kappa} = 1- \int_{0}^{\sigma^2_u} \textnormal{exp}(\frac{\lambda}{\rho^2 m + \sigma^2_u}) (\frac{\lambda}{m^2})\textnormal{exp}(\frac{-\lambda}{m}) \d m.
\end{align}
\end{Lemma}
\begin{proof} 
See Appendix \ref{app_2}. 
\end{proof}
Note that for stable systems, $M_\textnormal{max} = \frac{\sigma^2_u}{1- \rho^2}$ and not surprisingly $\rho^2 M_\textnormal{max} + \sigma^2_u = M_\textnormal{max}$. Then we get 
\begin{align}
\kappa^{b}_{l} & =\frac{1}{\left( a_{\kappa} \textnormal{exp}(\frac{-\lambda}{M_\textnormal{max}}) + \textnormal{exp}(-\frac{\lambda}{\sigma^2_u})\right) }
\end{align}
For unstable systems we have $M_\textnormal{max} \rightarrow \infty$, and as a result \begin{align}
\kappa^{\infty}_{l} & =\frac{1}{\left( a_{\kappa} + \textnormal{exp}(-\frac{\lambda}{\sigma^2_u})\right) }
\end{align}
To show how tight the bounds are, we have plotted the simulated pdf and two approximates using the bounds for $\kappa$ in Fig. \ref{fig_2}, given that $\sigma^2_u=\sigma^2_v =1$, $\lambda = 0.25$ (SNR = 6dB), and $\rho = 0.95$. 
 
With Lemma \ref{lemma:II} at hand, we are now ready to present the bounds for $P_\textnormal{out}(M_\textnormal{th})$. We then show that the bounds are tight for the high SNR regime, i.e. $\lambda \rightarrow 0$. This is discussed in the next section. 
\begin{figure}
\centering
\includegraphics[scale = 0.65] {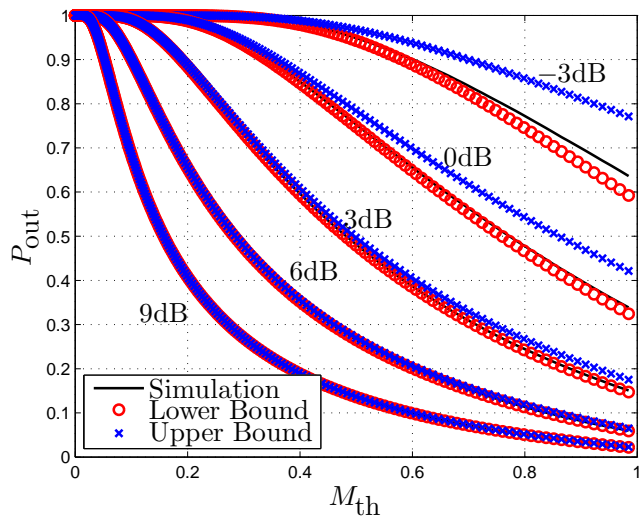}
\caption{$P_\textnormal{out}(M_\textnormal{th})$ and its upper and lower bounds for $\sigma^2_u = \sigma^2_v = 1$, $\rho = 0.95$, $\lambda = 2, 1, 0.5, 0.25, 0.125$ (SNR =-3, 0, 3, 6, 9 dB).}
\label{fig_3}
\end{figure}
\section {Bounds on Outage Probability for High SNR}\label{sec_IV}
In this section, we get upper and lower bounds for $P_\textnormal{out}(M_\textnormal{th})$. We show that for a given non-zero $M_\textnormal{th}$, EOP decreases with inverse of average channel SNR. We then show that the bounds are tight for the high SNR regime. 

As defined before, $P_\textnormal{out}(M_\textnormal{th})$ is given by
\begin{align}   
P_\textnormal{out} (M_\textnormal{th})= \int_{M_\textnormal{th}}^{M_\textnormal{max}} f_M(M) \d M.
\end{align}
For $M\leqslant \sigma^2_u$, we get 
\begin{align}   
P_\textnormal{out}(M_\textnormal{th}) = \int_{M_\textnormal{th}}^{M_\textnormal{max}} \frac{\kappa  \lambda}{M^2} \textnormal{exp}(\frac{-\lambda}{M}) \d M = 1 - \kappa \textnormal{exp}(\frac{-\lambda}{M_\textnormal{th}}).
\end{align}
As shown in the previous section, $\kappa_l < \kappa < \kappa_u$. As a result, we get 
\begin{align}
1 - \kappa_u \textnormal{exp}(\frac{-\lambda}{M_\textnormal{th}}) < P_\textnormal{out} (M_\textnormal{th}) < 1 - \kappa_l \textnormal{exp}(\frac{-\lambda}{M_\textnormal{th}}), 
\end{align}
which gives us an upper bound and a lower bound for $P_\textnormal{out}(M_\textnormal{th})$. Figure \ref{fig_3} depicts the outage probability and the bounds for the case when $\sigma^2_u=\sigma^2_v =1$, $\lambda = 2, 1, 0.5, 0.25, 0.125$ (SNR =-3, 0, 3, 6, 9 dB), and $\rho = 0.95$ and for $M\leqslant \sigma^2_u$. 

As seen in Fig. \ref{fig_3}, a smaller $\lambda$ yields a smaller outage probability. It is interesting to see how the increase in the SNR, i.e. a decrease in the value of $\lambda$, will lead to a lower outage probability. Also, we will show that the bounds are tight for high SNR, i.e. $\lambda \rightarrow 0$. 
 \begin{figure}
\centering
\includegraphics[scale = 0.65] {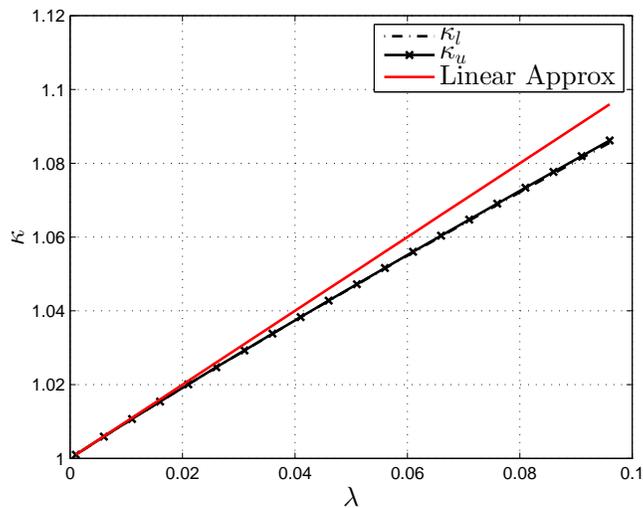}
\caption{$\kappa_l,\kappa_u$ as a function of $\lambda$  for $\sigma^2_u = \sigma^2_v = 1$, $\rho = 0.95$.}
\label{fig_4}
\end{figure}
\begin{Lemma}\label{lemma:III}
The upper and lower bounds for $P_\textnormal{out}(M_\textnormal{th})$ are tight for $\lambda\rightarrow 0$. 
\end{Lemma}
\begin{proof}
As shown in Appendix \ref{app_3}, we have that
\begin{align}
\lim_{\lambda \rightarrow 0} \kappa = \lim_{\lambda \rightarrow 0} \kappa_u = \lim_{\lambda \rightarrow 0} \kappa_l = 1,
\end{align}
which proves the lemma, because then the outage probability and the bounds will have the same values as  $\kappa, \kappa_u, \kappa_l$ converge to the same value. 
\end{proof}
It is also interesting to see how fast $\kappa$ converges to 1 for small values of $\lambda$ and for which values of $\lambda$, the upper and lower bound are approximately equal. This is depicted in Fig. \ref{fig_4}. Quite visibly, for values of $\lambda$ smaller than $0.01$ (SNR greater than 20dB) the upper and lower bounds for $\kappa$ are very close. Due to the fact that the bounds for $\kappa$ are tight, the bounds for $P_\textnormal{out}(M_\textnormal{th})$ are also tight. Even for the range of medium SNR depicted in Fig. \ref{fig_3}, it is quite visible that the upper and lower bounds for the outage probability are quite close to the one obtained from the simulation and that increasing the SNR improves their accuracy. However, the bounds, especially the upper bound lose their accuracy in the low SNR regime. This necessitates taking extra precautions if the bounds are to be used in applications prone to low SNR's. It is also quite visible from Fig. \ref{fig_4} that the linear approximation obtained from the Taylor series expansion of $\kappa$ as a function of $\lambda$ (see Appendix \ref{app_3}) is acurate.  
 \begin{figure}
\centering
\includegraphics[scale = 0.65] {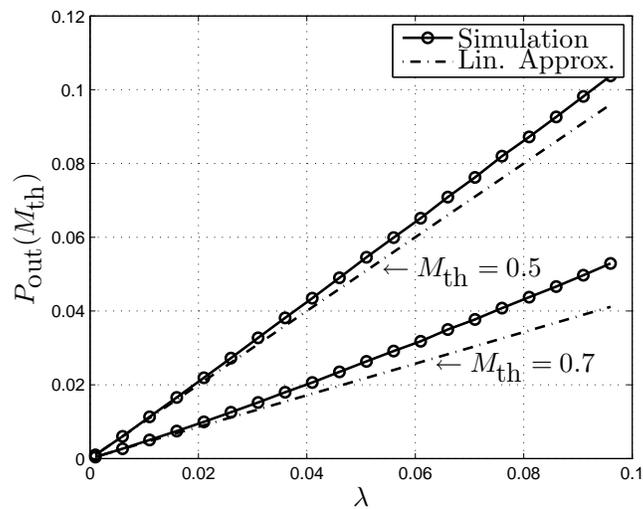}
\caption{$P_\textnormal{out}$ as a function of $\lambda$  for $\sigma^2_u = \sigma^2_v = 1$, $\rho = 0.95$ and its linear approximation}
\label{fig_5}
\end{figure} 
At this point we are ready to present the asymptotic behavior of the outage probability for the high SNR regime. This is presented in \cref{lemma:IV}. 
\begin{Lemma}\label{lemma:IV}
For the high SNR regime, $P_\textnormal{out}(M_\textnormal{th})$ decreases approximately linearly with $\lambda$. 
\end{Lemma}
\begin{proof}
We can use the Taylor series expansion of $\kappa$ around $\lambda = 0$ from Appendix \ref{app_3} to approximate the outage probability for the high SNR regime. Using the Taylor series expansion for $\kappa$ (from App. \ref{app_3}) and $\textnormal{exp}(\frac{-\lambda}{M_\textnormal{th}})$, we obtain
\begin{align}
P_\textnormal{out} (M_\textnormal{th})\nonumber & = 1- \kappa \textnormal{exp}(\frac{-\lambda}{M_\textnormal{th}})\\
& \nonumber  = 1 -   (  1+ \frac{\lambda}{\sigma^2_u}+ \mathcal{O}(\lambda^2)) ( 1 - \frac{\lambda}{M_\textnormal{th}} + \mathcal{O}(\lambda^2))\\
& = (\frac{1}{M_\textnormal{th}} -\frac {1}{\sigma^2_u})\lambda + \mathcal{O}(\lambda^2)
\end{align}
\end{proof}
For small $\lambda$, $\mathcal{O}(\lambda^2)$ vanishes faster than $\lambda$ and as a result we could claim $P_\textnormal{out}(M_\textnormal{th})$ is approximately a linear function of $\lambda$. The linear approximation is depicted in Fig. \ref{fig_5} for $\lambda \in [0.001, 0.01]$ ($\textnormal{SNR} \in [20\,\textnormal{dB}, 30\, \textnormal{dB}]$). The consequence of this linearity is that though increasing the channel SNR helps with outage probability, it does not help significantly and it may be beneficial to find a trade-off between power consumption and required outage probability for the application at hand.  
\section{Conclusions}
In this paper, a recursive integral equation approach was presented for finding the pdf of the instantaneous estimation error variance for MMSE estimation of scalar Gauss-Markov signals sent over fading channels. We also utilized the notion of estimation error outage as a means of characterizing the estimation accuracy. It was shown that the pdf can be written as a two-part function over the domain of instantaneous estimation error variance values. The first part of the pdf, i.e. the range up to the Gauss-Markov process variance also corresponding to higher outage probabilities, follows a closed-form non-recursive equation. As a result and for the case of i.i.d. Rayleigh fading channels, the outage probability can be approximated with a closed-form formula for the first part. Upper and lower bounds on the estimation error outage probability were also obtained to simplify characterization of estimation error outage. The presented bounds were shown to be tight when the SNR grows unbounded. In the end, it was shown that the outage decreases linearly with inverse of the SNR in the high SNR regime.  
\appendices
\section{Proof of Theorem 1}\label{app_1}
In order to prove Theorem \ref{theorem:I}, we refer to \cite{bougerol1993kalman} which considers Kalman filtering with random coefficients. In \cite{bougerol1993kalman}, a general vector state-space Kalman filter comprises the system model, which can be shown to include our system model as well. There and based on a contraction property of the Kalman filter, it is proven that the sequence of estimation error covariance matrices converges in law to a stationary process \cite[Theorem 2.4]{bougerol1993kalman}, given that some ergodicity conditions are met. In the following, we show in Lemma \ref{lemma:V} that those conditions are also met in our system model. As a result, the equivalent random variable in our case, i.e. $M(n)$ also converges in distribution and therefore $f_M(M)$ in fact exists. Then we prove in Lemma \ref{lemma:VI} why $f_M(M)$ can be obtained from \eqref{eqn_f_M_III}.  
\begin{Lemma}\label{lemma:V}
The instantaneous estimation error variance $M(n)$ converges in distribution.
\end{Lemma}
\begin{proof}
According to \cite[Theorem 2.4]{bougerol1993kalman}, three conditions are required for convergence of the estimation error covariance matrix of the Kalman filter with stochastic system parameters. Firstly, a hypothesis $\mathcal{H}$ (introduced in \cite[Section 2]{bougerol1993kalman}) should be satisfied. Secondly, it is required that the system is weakly observable and weakly controllable as defined in \cite{bougerol1993kalman}. Thirdly, certain (random) system parameters, specified later, should be integrable. 

For the first condition, it is mentioned in \cite[Section 2]{bougerol1993kalman} that a conditionally Gaussian system satisfies hypothesis $\mathcal{H}$. In our system, $u(n)$ and $v(n)$ are i.i.d. Gaussian random variables and also independent of $h(n)$. Therefore, our system is also conditionally Gaussian and satisfies the condition. 

For the second condition, we must show that our system is weakly controllable and weakly observable. Using the definition for weak controllability, we must show that
\begin{align}
\textnormal{Pr}\left(\sigma^2_u + \rho^2 \sigma^2_u + \rho^4 \sigma^2_u + \ldots + \rho^{2n} \sigma^2_u \neq 0\right) > 0, \label{eqn_con}\end{align}
which obviously holds as long as $\sigma^2_u > 0$. For weak controllability, we must show that 
\begin{align}   
\textnormal{Pr}\left(\rho^2 \gamma(1) + \rho^4 \gamma(2) + \ldots + \rho^{2n} \gamma(n) \neq 0\right) > 0.\label{eqn_obs} 
\end{align}
It is possible to show that \eqref{eqn_obs} holds for all channel distributions apart from the non-existent channel ($h(n) = 0$ for all $n$). Therefore, the second condition for convergence is also met.  

For the third condition to hold, we must have that the (random) variables $\log\log^+(\rho)$, $\log\log^+(\rho^{-1})$, $\log\log^+(\sigma^2_u)$ and $\log\log^+(h(1))$ are integrable (where $\log^+(x) = \max\left(\log(x),0\right)$), i.e. they have a well-defined expectation value (see e.g. \cite{williams1991probability} Chapter 13 for a definition of integrable random variables). Obviously, $\rho \neq 0$ and $\sigma^2_u > 0$ are deterministic variables. Therefore, they are integrable. $\log\log^+(h(1))$, is also integrable, given that $h(n)$ is defined as in Sec. \ref{Sec_II}. As a result, our system model satisfies all the prerequisites of Theorem 2.4 in \cite{bougerol1993kalman}. The consequence of the aforementioned theorem is that $M(n)$ converges in distribution (law) and therefore $f_M(M)$ exists. 
\end{proof}
\begin{Lemma}\label{lemma:VI}
The steady-state distribution of the random process $M(n)$ can be obtained from 
\begin{align}
& f_{M}(M) = \frac{1}{M^2} \int_{0}^{M_\textnormal{max}} f_{\gamma}\left(\frac{1}{M} - \frac{1}{\rho^2 m + \sigma^2_u}\right) f_{M} (m) \d m.\label{eqn_f_M_III_appendix}
\end{align}
\end{Lemma}
\begin{proof}
We use the fact that $M(n)$ is a Markov process. This is due to the fact that $M(n)$ is determined by only $M(n-1)$ and $\gamma(n)$. We have also shown in Lemma \ref{lemma:V} that $M(n)$ converges in distribution. From the theory of Markov processes in \cite{hernandez2003markov} and utilizing the relationship between $f_{M(n-1)}(M)$ and $f_{M(n)}(M)$ in \eqref{eqn_f_M_II_a}, we can verify that $M(n)$ has a transition probability measure (function) equal to $\frac{1}{M^2} f_{\gamma}\left(\frac{1}{M} - \frac{1}{\rho^2 m + \sigma^2_u}\right)$. We can then refer to Theorem 2.3.5 (ii) in \cite{hernandez2003markov} and conclude that
\begin{align}
& f_{M}(M) = \frac{1}{M^2} \int_{0}^{M_\textnormal{max}} f_{\gamma}\left(\frac{1}{M} - \frac{1}{\rho^2 m + \sigma^2_u}\right) f_{M} (m) \d m,\label{eqn_f_M_III_appendix}
\end{align}
\end{proof}
as stated in \eqref{eqn_f_M_III}.
\begin{proof}[Proof of Theorem \ref{theorem:I}] We have shown in Lemma \ref{lemma:V} that $M(n)$ converges in law. We have also shown in Lemma \ref{lemma:IV} that the steady-state distribution, namely $f_M(M)$ can be obtained from \eqref{eqn_f_M_III_appendix}. The proof is then complete. 
\end{proof}
\section{Upper and lower bounds for $\kappa$}\label{app_2}
We begin to rewrite $f_M(M)$ in the following manner for simplicity
\begin{align} \label{eqn_main_app1}f_M(M) = \left\{
\begin{array}{ll}
 \frac{\kappa \lambda}{M^2} \textnormal{exp}(\frac{-\lambda}{M}) & M\leqslant \sigma^2_u\\
 g(M) & M> \sigma^2_u.
\end{array} \right.
\end{align}
We have that $f_M(M)$ is a pdf, therefore $\int_{0}^{M_\textnormal{max}} f_M(M) \d M =1$. As a result, we have that
\begin{align}
\nonumber 1&= \int_{0}^{M_\textnormal{max}} f_M(M)\d M = \int_{0}^{\sigma^2_u} \frac{\kappa \lambda}{M^2} \textnormal{exp}(\frac{-\lambda}{M}) + \int_{\sigma^2_u}^{M_\textnormal{max}} g(M) \d M\\
& \nonumber =  \kappa \textnormal{exp}(\frac{-\lambda}{M})\Bigg|_{0}^{\sigma^2_u} + \int_{\sigma^2_u}^{M_\textnormal{max}} g(M) \d M\\
& =  \kappa \textnormal{exp}(\frac{-\lambda}{\sigma^2_u}) +  \int_{\sigma^2_u}^{M_\textnormal{max}} g(M) \d M,
\end{align}
which gives 
\begin{align}
\label{eqn_integg}\int_{\sigma^2_u}^{M_\textnormal{max}} g(M) \d M = 1 -  \kappa \textnormal{exp}(\frac{-\lambda}{\sigma^2_u}).
\end{align}
Now we take 
\begin{align*}
\sigma^2_u < m < M_\textnormal{max}.
\end{align*}
Then we have
\begin{align*}
(\rho^2 + 1)\sigma^2_u < \rho^2 m + \sigma^2_u< \rho^2 M_\textnormal{max} + \sigma^2_u
\end{align*}
and 
\begin{align}
 \textnormal{exp}\left(\frac{\lambda}{\rho^2 M_\textnormal{max} + \sigma^2_u}\right) < \textnormal{exp}\left(\frac{\lambda}{\rho^2 m + \sigma^2_u}\right) < \textnormal{exp}\left(\frac{\lambda}{(\rho^2 + 1)\sigma^2_u}\right).
\end{align}

Now we have that
\begin{align}
&\int_{\sigma^2_u}^{M_\textnormal{max}}\textnormal{exp}\left(\frac{\lambda}{\rho^2 m + \sigma^2_u}\right) g(m) \d m >  \int_{\sigma^2_u}^{M_\textnormal{max}}\textnormal{exp}\left(\frac{\lambda}{\rho^2 M_\textnormal{max} + \sigma^2_u}\right) g(m)\d m\label{eqn_lb}\\
& \int_{\sigma^2_u}^{M_\textnormal{max}}\textnormal{exp}\left(\frac{\lambda}{\rho^2 m + \sigma^2_u}\right) g(m) \d m<  \int_{\sigma^2_u}^{M_\textnormal{max}}\textnormal{exp}\left(\frac{\lambda}{(\rho^2+1) \sigma^2_u}\right) g(m)\d m.\label{eqn_ub}
\end{align}
Next, if we use the definition of $f_M(M)$ for $M\leqslant \sigma^2_u$, we obtain the following
\begin{align}
f(M) &\nonumber = \frac{\lambda}{M^2}\textnormal{exp}(\frac{-\lambda}{M}) \int_{0}^{M_\textnormal{max}} \textnormal{exp}(\frac{\lambda}{\rho^2 m + \sigma^2_u}) f_M(m)\d m\\
& = \frac{\kappa \lambda}{M^2}\textnormal{exp}(\frac{-\lambda}{M})\label{eqn_leftside}\\
& = \frac{\lambda}{M^2}\textnormal{exp}(\frac{-\lambda}{M})  \left( \int_{0}^{\sigma^2_u} \textnormal{exp}(\frac{\lambda}{\rho^2 m + \sigma^2_u}) \left(\frac{\kappa \lambda}{m^2} \right) \textnormal{exp}(\frac{-\lambda}{m}) \d m\right.\\
 & \left.+ \int_{\sigma^2_u} ^{M_\textnormal{max}}\textnormal{exp}(\frac{\lambda}{\rho^2 m + \sigma^2_u}) g(m) \d m\right),\label{eqn_rightside}
\end{align}
from which by equating \eqref{eqn_leftside} and \eqref{eqn_rightside} and removing common terms on both sides, we deduce that
\begin{align}
\kappa \nonumber & = \kappa \int_{0}^{\sigma^2_u} \textnormal{exp}(\frac{\lambda}{\rho^2 m + \sigma^2_u}) \left(\frac{ \lambda}{m^2}\right) \textnormal{exp}(\frac{-\lambda}{m}) \d m\\
& + \int_{\sigma^2_u} ^{M_\textnormal{max}}\textnormal{exp}(\frac{\lambda}{\rho^2 m + \sigma^2_u}) g(m)\d m.
\end{align}
And then we obtain 
\begin{align}\label{eqn_kappa_app}
\kappa = \frac{\int_{\sigma^2_u} ^{M_\textnormal{max}}\textnormal{exp}(\frac{\lambda}{\rho^2 m + \sigma^2_u}) g(m)\d m}{1- \int_{0}^{\sigma^2_u} \textnormal{exp}(\frac{\lambda}{\rho^2 m + \sigma^2_u}) \left(\frac{ \lambda}{m^2}\right) \textnormal{exp}(\frac{-\lambda}{m}) \d m}.
\end{align}
Now by letting 
\begin{align}
a_{\kappa} = 1- \int_{0}^{\sigma^2_u} \textnormal{exp}(\frac{\lambda}{\rho^2 m + \sigma^2_u}) \left(\frac{ \lambda}{m^2}\right) \textnormal{exp}(\frac{-\lambda}{m}) \d m
\end{align}
and combining \eqref{eqn_lb} into \eqref{eqn_kappa_app} while using \eqref{eqn_integg}, we get 
\begin{align}
\kappa a_{\kappa} & \nonumber > \int_{\sigma^2_u}^{M_\textnormal{max}}\textnormal{exp}(\frac{\lambda}{\rho^2 M_\textnormal{max} + \sigma^2_u}) g(m)\d m\\
& > \textnormal{exp}(\frac{\lambda}{\rho^2 M_\textnormal{max} + \sigma^2_u})  \int_{\sigma^2_u}^{M_\textnormal{max}}g(m)\d m\\
& > \textnormal{exp}(\frac{\lambda}{\rho^2 M_\textnormal{max} + \sigma^2_u}) (1 -  \kappa \textnormal{exp}(\frac{-\lambda}{\sigma^2_u})),
\end{align}
which leads to 
\begin{align}
\kappa  > \frac{1}{(a_{\kappa} \textnormal{exp}(\frac{-\lambda}{\rho^2 M_\textnormal{max} + \sigma^2_u}) + \textnormal{exp}(\frac{-\lambda}{\sigma^2_u}))} .
\end{align}
So, we finally get
\begin{align}
\kappa_l =  \frac{1}{\left(a_{\kappa} \textnormal{exp}\left(\frac{-\lambda}{\rho^2 M_\textnormal{max} + \sigma^2_u}\right) + \textnormal{exp}(\frac{-\lambda}{\sigma^2_u})\right)}.
\end{align}
The same procedure also holds for $\kappa_u$ by integrating \eqref{eqn_ub} into \eqref{eqn_kappa} while using \eqref{eqn_integg}. We then get 
\begin{align}
\kappa_{u} & =\frac{1}{\left( a_{\kappa} \textnormal{exp}\left(\frac{-\lambda}{\sigma^2_u(1+\rho^2)}\right) + \textnormal{exp}(-\frac{\lambda}{\sigma^2_u})\right)}.
\end{align}
\section{High SNR Limits for $\kappa, \kappa_u, \kappa_l$}\label{app_3}
In this section we show that $\kappa, \kappa_u,\kappa_l$ converge to 1 in the high SNR regime, i.e. in the limit of $\lambda\rightarrow 0$. We also get the Taylor series expansion for $\kappa$ to accommodate for the high SNR linear approximation for the outage probability in Lemma \ref{lemma:IV}.    

We have 
\begin{align}
\kappa_{u} & =\frac{1}{\left( a_{\kappa} \textnormal{exp}\left(\frac{-\lambda}{\sigma^2_u(1+\rho^2)}\right) + \textnormal{exp}(-\frac{\lambda}{\sigma^2_u})\right) }.
\end{align}
For finite $\sigma^2_u$, the condition $\lambda \rightarrow 0 $ can be extended to $\lambda / \sigma^2_u \rightarrow 0$. We make this assumption to simplify the calculations. Assume $\lambda = \alpha \sigma^2_u$, then 
\begin{align}
\kappa_{u} & =\frac{1}{\left( a_{\kappa}(\alpha) \textnormal{exp}\left(\frac{-\alpha}{(1+\rho^2)}\right) + \textnormal{exp}(-\alpha)\right) }.
\end{align}
For finite $\sigma^2_u$, the condition $\lambda \rightarrow 0 $ will be equal to $\alpha \rightarrow 0$. We can then see that 
\begin{align}
\lim_{\alpha \rightarrow 0} \kappa_u(\alpha) = \frac{1}{ 1+ \lim_{\alpha \rightarrow 0} a_\kappa(\alpha)}.
\end{align}
In the following, we show that $\lim_{\alpha \rightarrow 0}a_\kappa(\alpha) = 0$. 
We have 
\begin{align}
a_{\kappa} (\alpha)& \nonumber = 1- \int_{0}^{\sigma^2_u} \textnormal{exp}(\frac{\lambda}{\rho^2 m + \sigma^2_u}) \left(\frac{ \lambda}{m^2}\right) \textnormal{exp}(\frac{-\lambda}{m}) \d m\\
& \label{eqn_akappa}= 1- \int_{0}^{1} \textnormal{exp}(\frac{\alpha}{1 + \rho^2 v }) \left(\frac{ \alpha}{v^2}\right) \textnormal{exp}(\frac{-\alpha}{v}) \d v,
\end{align}
where we made the change of variable $v = \frac{m}{\sigma^2_u}$. Now take $a_\kappa (\alpha) = 1 -I(\alpha)$, where  
\begin{align}
I(\alpha) & = \int_{0}^{1} \textnormal{exp}(\frac{\alpha}{1 + \rho^2 v }) \left(\frac{ \alpha}{v^2}\right) \textnormal{exp}(\frac{-\alpha}{v}) \d v\\
& =\nonumber  \textnormal{exp}(\frac{\alpha}{1 + \rho^2 v }) \textnormal{exp}(\frac{-\alpha}{v})\Bigg |_{0}^{1} \\
& - \int_{0}^{1} \textnormal{exp}(\frac{-\alpha}{v})  \textnormal{exp}(\frac{\alpha}{1 + \rho^2 v }) \left(\frac{ -\alpha \rho^2}{(1+\rho^2 v)^2}\right) \d v\\
& = \nonumber \textnormal{exp}(\frac{\alpha}{1 + \rho^2 })\textnormal{exp}(-\alpha)\\
& + \label{eqn_intg_alpha}\rho^2 \alpha  \int_{0}^{1} \textnormal{exp}(\frac{-\alpha}{v})  \textnormal{exp}(\frac{\alpha}{1 + \rho^2 v }) \frac{ 1}{(1+\rho^2 v)^2} \d v.
\end{align}
Now, because all of the functions $\textnormal{exp}(\frac{-\alpha}{v}), \textnormal{exp}(\frac{\alpha}{1 + \rho^2 v }), \frac{1}{(1+\rho^2 v)^2}$ are finite for $v \in [0, 1]$, then the integral term in \eqref{eqn_intg_alpha} is also finite. As a result, $\lim_{\alpha \rightarrow 0} I(\alpha) = 1$ and $\lim_{\alpha \rightarrow 0} \kappa_u(\alpha) = 1$. Similar results also hold for $\kappa_l^{\infty}, \kappa_l^{b}$. 

To get the limiting behavior for $\kappa$ when $\lambda \rightarrow 0$, we use the original definition for $\kappa$, i.e. 
\begin{align}
\kappa = \int_{0}^{M_\textnormal{max}}\textnormal{exp}(\frac{\lambda}{\rho^2 m + \sigma^2_u}) f_M(m) \d m
\end{align}
to obtain the aforementioned limit. As a prerequisite for lemmas \ref{lemma:III} and \ref{lemma:IV}, we also need the Taylor series expansion for $\kappa$ around the point $\lambda =0$, which is done in the following. 

We begin by first showing that the cdf of IEV, i.e. $F_M(M)$ approaches the step function when $\lambda \rightarrow 0$. We have that
\begin{align}
F_M(M) = 1 - P_\textnormal{out}(M) =  \kappa \textnormal{exp}(\frac{-\lambda}{M}).
\end{align}
Now for any $M>0$, we have 
\begin{align}
\lim_{\lambda \rightarrow 0} F_M(M) & = \lim_{\lambda \rightarrow 0} \kappa \textnormal{exp}(\frac{-\lambda}{M})  \nonumber\\
& = \lim_{\lambda \rightarrow 0} \kappa. \label{F_M_kappa}
\end{align}

Now we see that
\begin{align}
\lim_{\lambda \rightarrow 0} \kappa & = \lim_{\lambda \rightarrow 0} \int_{0}^{M_\textnormal{max}} \textnormal{exp}(\frac{\lambda}{\rho^2 m + \sigma^2_u}) f_M(m) \d m\nonumber \\
& = \lim_{\lambda \rightarrow 0} \int_{0}^{M_\textnormal{max}} \sum_{l=0}^{\infty} \frac{\lambda^l}{l!}\frac{1}{(\rho^2 m + \sigma^2_u)^l} f_M(m) \d m\nonumber \\
& = \lim_{\lambda \rightarrow 0}  \int_{0}^{M_\textnormal{max}} f_M(m) \d m \nonumber \\
& +\lim_{\lambda \rightarrow 0} \int_{0}^{M_\textnormal{max}} \sum_{l=1}^{\infty} \frac{\lambda^l}{l!}\frac{1}{(\rho^2 m + \sigma^2_u)^l} f_M(m) \d m\nonumber \\
& = \lim_{\lambda \rightarrow 0} 1 +  \lim_{\lambda \rightarrow 0} \int_{0}^{M_\textnormal{max}} \sum_{l=1}^{\infty} \frac{\lambda^l}{l!}\frac{1}{(\rho^2 m + \sigma^2_u)^l} f_M(m) \d m\nonumber \\
& = 1 +  \lim_{\lambda \rightarrow 0} \int_{0}^{M_\textnormal{max}} \sum_{l=1}^{\infty} \frac{\lambda^l}{l!}\frac{1}{(\rho^2 m + \sigma^2_u)^l} f_M(m) \d m, \nonumber \\
\end{align}
but for $m \geqslant 0$, we have that $\frac{1}{\rho^2 m + \sigma^2_u} \leqslant \frac{1}{\sigma^2_u}$. As a result we get
\begin{align}
& \lim_{\lambda \rightarrow 0} \int_{0}^{M_\textnormal{max}} \sum_{l=1}^{\infty} \frac{\lambda^l}{l!}\frac{1}{(\rho^2 m + \sigma^2_u)^l} f_M(m) \d m \leqslant  \nonumber \\
& \lim_{\lambda \rightarrow 0} \sum_{l=1}^{\infty} \frac{\lambda^l}{l!}\frac{1}{(\sigma^2_u)^l}\int_{0}^{M_\textnormal{max}}f_M(m) \d m, 
\end{align}
and thus
\begin{align}
& \lim_{\lambda \rightarrow 0} \int_{0}^{M_\textnormal{max}} \sum_{l=1}^{\infty} \frac{\lambda^l}{l!}\frac{1}{(\rho^2 m + \sigma^2_u)^l} f_M(m) \d m \leqslant  \nonumber \\
& \lim_{\lambda \rightarrow 0} \sum_{l=1}^{\infty} \frac{\lambda^l}{l!}\frac{1}{(\sigma^2_u)^l},
\end{align}
but 
\begin{align}
\lim_{\lambda \rightarrow 0} \sum_{l=1}^{\infty} \frac{\lambda^l}{l!}\frac{1}{(\sigma^2_u)^l} = \lim_{\lambda \rightarrow 0} (\textnormal{exp}(\frac{\lambda}{\sigma^2_u}) -1)= 0,
\end{align}
therefore 
\begin{align}
 \lim_{\lambda \rightarrow 0} \int_{0}^{M_\textnormal{max}} \sum_{l=1}^{\infty} \frac{\lambda^l}{l!}\frac{1}{(\rho^2 m + \sigma^2_u)^l} f_M(m) \d m = 0,
\end{align}
and finally, $ \lim_{\lambda \rightarrow 0} \kappa =1$ as claimed before. In addition and from \eqref{F_M_kappa}, this result shows that $F_M(M)$ approaches the unit step function when $\lambda \rightarrow 0$ and as a result $f_M(M)$ approaches the Dirac's delta function when $\lambda \rightarrow 0$. With this assumption we have 
\begin{align}
\kappa & = \int_{0}^{M_\textnormal{max}} \textnormal{exp}\left(\frac{\lambda}{\rho^2 m +\sigma^2_u}\right)f_M(m) \d m\nonumber\\
& =  \int_{0}^{M_\textnormal{max}}  \sum_{l=0}^{\infty} \frac{\lambda^l}{l!}\frac{1}{(\rho^2 m + \sigma^2_u)^l} f_M(m)\d m\nonumber \\
& = \int_{0}^{M_\textnormal{max}} \sum_{l=0}^{\infty} \frac{\lambda^l}{l!}\frac{1}{(\rho^2 m + \sigma^2_u)^{l}} f_M(m)\d m\nonumber\\
& = \sum_{l=0}^{\infty} \frac{\lambda^l}{l!}\int_{0}^{M_\textnormal{max}} \frac{1}{(\rho^2 m + \sigma^2_u)^{l}} f_M(m)\d m\nonumber\\
& = \sum_{l=0}^{\infty} \frac{\lambda^l}{l!}\int_{0}^{M_\textnormal{max}} \frac{1}{(\rho^2 m + \sigma^2_u)^{l}} \delta(m)\d m\nonumber\\
& = \sum_{l=0}^{\infty} \frac{\lambda^l}{l!} \frac{1}{(\sigma^2_u)^{l}},
\end{align}
which is the Taylor series expansion for $\kappa$ around $\lambda = 0$ to be used in lemmas \ref{lemma:III} and \ref{lemma:IV}. 

\bibliographystyle{ieeetran}
\bibliography{references}

\begin{thebibliography}{10}
\providecommand{\url}[1]{#1}
\csname url@samestyle\endcsname
\providecommand{\newblock}{\relax}
\providecommand{\bibinfo}[2]{#2}
\providecommand{\BIBentrySTDinterwordspacing}{\spaceskip=0pt\relax}
\providecommand{\BIBentryALTinterwordstretchfactor}{4}
\providecommand{\BIBentryALTinterwordspacing}{\spaceskip=\fontdimen2\font plus
\BIBentryALTinterwordstretchfactor\fontdimen3\font minus
  \fontdimen4\font\relax}
\providecommand{\BIBforeignlanguage}[2]{{%
\expandafter\ifx\csname l@#1\endcsname\relax
\typeout{** WARNING: IEEEtran.bst: No hyphenation pattern has been}%
\typeout{** loaded for the language `#1'. Using the pattern for}%
\typeout{** the default language instead.}%
\else
\language=\csname l@#1\endcsname
\fi
#2}}
\providecommand{\BIBdecl}{\relax}
\BIBdecl

\bibitem{zhang2001stability}
W.~Zhang, M.~S. Branicky, and S.~M. Phillips, ``{Stability of networked control
  systems},'' \emph{Control Systems, IEEE}, vol.~21, no.~1, pp. 84--99, 2001.

\bibitem{lihua2011control}
X.~Lihua, ``{Control over communication networks: Trend and challenges in
  integrating control theory and information theory},'' in \emph{2011 30th
  Chinese Control Conference (CCC)}.\hskip 1em plus 0.5em minus 0.4em\relax
  IEEE, 2011, pp. 35--39.

\bibitem{antsaklis2007special}
P.~Antsaklis and J.~Baillieul, ``{Technology of networked control systems},''
  \emph{Special Issue of Proceedings of the IEEE}, vol.~95, no.~1, 2007.

\bibitem{Stankovic:2008hk}
J.~A. Stankovic, ``{Wireless Sensor Networks},'' \emph{{Computer}}, vol.~41,
  no.~10, pp. 92--95, Oct. 2008.

\bibitem{bougerol1993kalman}
P.~Bougerol, ``Kalman filtering with random coefficients and contractions,''
  \emph{SIAM Journal on Control and Optimization}, vol.~31, no.~4, pp.
  942--959, 1993.

\bibitem{shi2009kalman}
L.~Shi, L.~Xie, and R.~M. Murray, ``{Kalman filtering over a packet-delaying
  network: A probabilistic approach},'' \emph{Automatica}, vol.~45, no.~9, pp.
  2134--2140, 2009.

\bibitem{Shi:2012en}
L.~Shi and L.~Xie, ``{Optimal sensor power scheduling for state estimation of
  Gauss--Markov systems over a packet-dropping network},'' \emph{{IEEE Trans.
  Signal Process.}}, vol.~60, no.~5, pp. 2701--2705, 2012.

\bibitem{rohr2011kalman}
E.~Rohr, D.~Marelli, and M.~Fu, ``{Kalman filtering for a class of degenerate
  systems with intermittent observations},'' in \emph{2011 50th IEEE Conference
  on Decision and Control and European Control Conference (CDC-ECC)}.\hskip 1em
  plus 0.5em minus 0.4em\relax IEEE, 2011, pp. 2422--2427.

\bibitem{you2011mean}
K.~You, M.~Fu, and L.~Xie, ``{Mean square stability for Kalman filtering with
  Markovian packet losses},'' \emph{Automatica}, vol.~47, no.~12, pp.
  2647--2657, 2011.

\bibitem{quevedo2012kalman}
D.~E. Quevedo, A.~Ahl{\'e}n, A.~S. Leong, and S.~Dey, ``{On Kalman filtering
  over fading wireless channels with controlled transmission powers},''
  \emph{Automatica}, vol.~48, no.~7, pp. 1306--1316, 2012.

\bibitem{Subramanian:2005ca}
A.~Subramanian and A.~H. Sayed, ``{Joint rate and power control algorithms for
  wireless networks},'' \emph{{IEEE Trans. Signal Process.}}, vol.~53, no.~11,
  pp. 4204--4214, Nov. 2005.

\bibitem{Zhu:2009dg}
H.~Zhu, I.~D. Schizas, and G.~B. Giannakis, ``{Power-efficient dimensionality
  reduction for distributed channel-aware Kalman tracking using WSNs},''
  \emph{{IEEE Trans. Signal Process.}}, vol.~57, no.~8, pp. 3193--3207, Aug.
  2009.

\bibitem{wonham1968matrix}
W.~M. Wonham, ``{On a matrix Riccati equation of stochastic control},''
  \emph{SIAM Journal on Control}, vol.~6, no.~4, pp. 681--697, 1968.

\bibitem{wang1999stability}
Y.~Wang and L.~Guo, ``{On stability of random {Riccati} equations},''
  \emph{Science in China Series E: Technological Sciences}, vol.~42, no.~2, pp.
  136--148, 1999.

\bibitem{xie2008stability}
L.~Xie and L.~Xie, ``{Stability of a random Riccati equation with Markovian
  binary switching},'' \emph{{IEEE} Trans. Autom. Control}, vol.~53, no.~7, pp.
  1759--1764, 2008.

\bibitem{Moayedi:2010tV}
M.~Moayedi, Y.~K. Foo, and Y.~C. Soh, ``Adaptive {Kalman} filtering in
  networked systems with random sensor delays, multiple packet dropouts and
  missing measurements,'' \emph{{IEEE Trans. Signal Process.}}, vol.~58, no.~3,
  pp. 1577--1588, Jan. 2010.

\bibitem{kar2012kalman}
S.~Kar, B.~Sinopoli, and J.~M. Moura, ``{Kalman filtering with intermittent
  observations: Weak convergence to a stationary distribution},'' \emph{{IEEE}
  Trans. Autom. Control}, vol.~57, no.~2, pp. 405--420, 2012.

\bibitem{vakili2008stieltjes}
A.~Vakili and B.~Hassibi, ``{A Stieltjes transform approach for analyzing the
  RLS adaptive filter},'' in \emph{46th Annual Allerton Conference on
  Communication, Control, and Computing}.\hskip 1em plus 0.5em minus
  0.4em\relax IEEE, 2008, pp. 432--437.

\bibitem{vakili2008riccati}
------, ``{A Stieltjes transform approach for studying the steady-state
  behavior of random Lyapunov and Riccati recursions},'' in \emph{47th IEEE
  Conference on Decision and Control (CDC)}.\hskip 1em plus 0.5em minus
  0.4em\relax IEEE, 2008, pp. 453--458.

\bibitem{Dey2009kalman}
S.~Dey, A.~S. Leong, and J.~S. Evans, ``{Kalman filtering with faded
  measurements},'' \emph{Automatica}, vol.~45, no.~10, pp. 2223--2233, 2009.

\bibitem{caire2007distortion}
G.~Caire and K.~Narayanan, ``{On the distortion SNR exponent of hybrid
  digital--analog space--time coding},'' \emph{{IEEE Trans. Inf. Theory}},
  vol.~53, no.~8, pp. 2867--2878, 2007.

\bibitem{peng2010distortion}
L.~Peng and A.~Guillen~i Fabregas, ``{Distortion outage probability in MIMO
  block-fading channels},'' in \emph{2010 IEEE International Symposium on
  Information Theory Proceedings (ISIT)}, 2010, pp. 2223--2227.

\bibitem{parseh2014dist}
R.~Parseh and K.~Kansanen, ``{On estimation error outage for scalar
  Gauss-Markov processes sent over fading channels},'' in \emph{22nd European
  Signal Processing Conference, (EUSIPCO)}, Lisbon, Portugal, Sep. 2014.

\bibitem{kailath2000linear}
T.~Kailath, A.~H. Sayed, and B.~Hassibi, \emph{{Linear estimation}}.\hskip 1em
  plus 0.5em minus 0.4em\relax Prentice Hall New Jersey, 2000, vol.~1.

\bibitem{papoulis2002probability}
A.~Papoulis and S.~U. Pillai, \emph{Probability, random variables, and
  stochastic processes}, 4th~ed.\hskip 1em plus 0.5em minus 0.4em\relax
  McGraw-Hill Europe, Jan. 2002.

\bibitem{williams1991probability}
D.~Williams, \emph{Probability with martingales}.\hskip 1em plus 0.5em minus
  0.4em\relax Cambridge university press, 1991.

\bibitem{hernandez2003markov}
O.~Hern{\'a}ndez-Lerma and J.~B. Lasserre, \emph{Markov chains and invariant
  probabilities}.\hskip 1em plus 0.5em minus 0.4em\relax Springer, 2003.

\end{thebibliography}
\end{document}